%% file: main.tex
\pdfoutput=1
\documentclass[11pt,letterpaper]{article}
\usepackage[utf8]{inputenc}

\usepackage{lmodern}
\usepackage[T1]{fontenc}
\input{preamble.tex}
\usepackage{authblk}

\usepackage{xcolor}

\ifdefined\ShowComment
\newcommand{\danupon}[1]{{\bf \color{green} DANUPON: #1}}
\newcommand{\sagnik}[1]{{\bf \color{blue} SAGNIK: #1}}
\newcommand{\jan}[1]{{\bf \color{red} Jan: #1}}
\newcommand{\blikstad}[1]{{\bf \color{magenta} [Joakim: #1]}}
\else
\newcommand{\danupon}[1]{}
\newcommand{\sagnik}[1]{}
\newcommand{\jan}[1]{}
\newcommand{\blikstad}[1]{}
\fi

\newcommand{\LightEdges}{\textsc{LightEdges}}
\newcommand{\NotVisited}{\textsc{NotVisited}}

\title{Breaking the Quadratic Barrier for Matroid Intersection}

\author[1]{Joakim Blikstad}
\author[1]{Jan van den Brand}
\author[1]{Sagnik Mukhopadhyay}
\author[1]{Danupon Nanongkai}
\affil[1]{KTH Royal Institute of Technology, Sweden}
\affil[ ]{\tt\{blikstad,janvdb,sagnik,danupon\}@kth.se}
\date{}

\begin{document}

\begin{titlepage}
	\maketitle \pagenumbering{roman}
	\input{abstract}

	\newpage
	\setcounter{tocdepth}{2}
	\tableofcontents
\end{titlepage}

\newpage
\pagenumbering{arabic}
\input{intro}

\section{Preliminaries} \label{sec:prelims}

\paragraph{Matroid.} A \textit{matroid} is a combinatorial object defined by the tuple $\cM= (V, \cI)$, where the ground set $V$ is a finite set of elements and $\cI \subseteq 2^{V}$ is a non-empty family of subsets (denoted as the \textit{independent sets}) of the ground set $V$, such that the following properties hold: \begin{enumerate}
    \item \textbf{Downward closure:} If $S \in \cI$, then any subset $S' \subset S$ (including the empty set) is also in $\cI$,
    
    \item \textbf{Exchange property:} For any two sets $S_1, S_2 \in \cI$ with $|S_1| < |S_2|$, there is an element $v \in S_2 \setminus S_1$ such that $S_1 \cup \{v\} \in \cI$.
\end{enumerate}

\paragraph{Matroid Intersection.} Given two matroids $\cM_1 = (V, \cI_1)$ and $\cM_2 = (V, \cI_2)$ defined on the same ground set $V$, the \emph{matroid intersection problem}
is finding a maximum cardinality common independent set $S\in \cI_1 \cap \cI_2$.
When discussing matroid intersection, we will denote by $r$ the size of such a maximum cardinality common independent set and by $n$ the size of the ground set $V$.

\paragraph{Exchange graph.} Consider two matroids $\cM_1 = (V, \cI_1)$ and $\cM_2 = (V, \cI_2)$ defined on the same ground set $V$. Let $S \in \cI_1 \cap \cI_2$ be a common independent set. The \textit{exchange graph} $G(S)$, w.r.t. to the common independent set $S\in\cI_1\cap \cI_2$, is defined to be a directed bipartite graph where the two sides of the bipartition are $S$ and $\bar S = V\setminus S$. Moreover, there are two additional special vertices $s$ and $t$ (that are not included in either $S$ or $\bar S$) which have directed edges incident on them \textit{only} from $\bar S$. The directed edges (or arcs) are interpreted as follows: \begin{enumerate}
    \item Any edge of the form $(s,v)$ for $v \in \bar S$ implies that $S \cup \{v\}$ is an independent set in $\cM_1$.
    \item Similarly, any edge of the form $(v, t)$ for $v \in \bar S$ implies that $S \cup \{v\}$ is an independent set in $\cM_2$.
    \item Any edge of the form $(u, v) \in S \times \bar S$ implies that $(S \setminus \{u\}) \cup \{v\}$ is an independent set in $\cM_1$.
    \item Similarly, any edge of the form $(v, u) \in \bar S \times S$ implies that $(S \setminus \{u\}) \cup \{v\}$ is an independent set in $\cM_2$.
\end{enumerate}

We are interested in the notion of \textit{chordless} $(s,t)$-paths in $G(S)$ \cite[Section 2]{cunningham1986improved} which are defined next. For this definition, we consider a path as a sequence of vertices that take part in the path. A subsequence of a path is an ordered  subset of the vertices (not necessarily contiguous) of the path where the ordering respects the path ordering.

\begin{definition}
 An $(s,t)$-path $p$ is \emph{chordless} if there is no proper subsequence of $p$ which is also an $(s,t)$-path.
A chordless path in the exchange graph $G(S)$ is sometimes called an \emph{augmenting path}. 
\end{definition}

\begin{claim}[Augmenting path] \label{clm:exg-graph}
Consider a chordless path $p$ from $s$ to $t$ in $G(S)$ (if it exists), and let $V(p)$ be the elements of the ground set (or, equivalently, vertices in the exchange graph excluding $s$ and $t$) that take part in the path $p$. Then $S \triangle V(p)$ is a common independent set of $\cM_1$ and $\cM_2$.
\end{claim}

If we examine the set $S \triangle V(p)$ obtained from Claim \ref{clm:exg-graph}, it is clear that the number of elements added to the set $S$ is one more than the number of elements removed from $S$. This observation immediately gives the following corollary,
and shows the importance of the notion of exchange graphs.

\begin{corollary}
The size of the largest common independent set of $\cM_1$ and $\cM_2$ is at least $|S| + 1$ if and only if $t$ is reachable from $s$ in $G(S)$.
\end{corollary}

It is useful to note that the shortest $(s,t)$-path in $G(S)$ is always chordless.
Many combinatorial matroid intersection algorithms thus focus on finding
shortest $(s,t)$-paths.
The following claim relating the distance from $s$ to $t$ in $G(S)$ and
the size of $S$ is useful for approximation algorithms for matroid intersection.

\begin{claim}[\cite{cunningham1986improved}] \label{clm:dists}
If the length of the shortest $(s,t)$-path in $G(S)$ is at least $d$, then 
$|S| \ge (1-O(\frac{1}{d}))r$, where $r$ is the size of the largest common independent set.
\end{claim}

\paragraph{Matroid query oracles.} There are two primary models of query oracles associated with the matroid theory: (i) the independence query oracle, and (ii) the rank query oracle. The independence query oracle, given a set $S \subseteq V$ of a matroid $\cM$, outputs 1 iff $S$ is an independent set of $\cM$ (i.e., iff $S \in \cI$). The rank query oracle, given a set $S \subseteq V$, outputs the rank of $S$, $\rk_\cM(S) \overset{def}{=} \max_{T \subseteq S : T \in \cI}|T|$, i.e., the size of the largest independent set contained in $S$. Clearly, if $S$ itself is an independent set, then $\rk_\cM(S) = |S|$. Hence, a rank query oracle is at least as powerful as the independence query oracle. In this work, we are however interested primarily in the independence query oracle model. Next, we state two claims regarding the independence query oracle that we use in the paper.

%
%
\begin{claim}[Edge discovery] \label{clm:is-edge}
    By issuing one independence query each, we can find out
    \begin{enumerate}[label=(\roman*)]
    \item given a vertex $v\in \bar S$, whether $v$ is an out-neighbor of $s$; or
    \item given a vertex $v\in \bar S$, whether $v$ is an in-neighbor of $t$; or
    \item given a vertex $v\in \bar S$ and a subset $X\subseteq S$, whether there exists an edge from some vertex in $X$ to $v$; or
    \item given a vertex $v\in \bar S$ and a subset $X\subseteq S$, whether there exists an edge from $v$ to some vertex in $X$.
\end{enumerate}
\end{claim}

Claim \ref{clm:is-edge} follows from observing that we can make the following kinds of independence queries:
(i-ii) whether $S \cup \{v\}$ is an independent set in $\cM_1$ respectively $\cM_2$, and
(iii-iv) whether $S \cup \{v\} \setminus X$ is an independent set in $\cM_1$ respectively $\cM_2$.
Note that these edge-discovery queries can simulate the neighborhood-queries in the reachability problem.

With these kinds of queries, we can perform a binary search to find an in-/out-neighbor
of $v\in \bar{S}$.
The following lemma is proven in \cite[Lemma 11]{chakrabarty2019faster} and also mentioned in \cite{nguyen2019note}. We skip the proof in the paper.

\begin{claim}[Binary search with independence/neighborhood queries, \cite{nguyen2019note,chakrabarty2019faster}] \label{clm:bin-search}
Consider a vertex $v \in \bar S$ and a subset $X\subseteq S\cup \{s,t\}$. By issuing $O(\log r)$ independence queries to $\cM_1$, we can find a vertex $u\in X$ such that there is an edge $(u,v)$ (i.e., $u$ is an in-neighbor of $v$), or otherwise determine that no such
edge exists. Similarly, by issuing $O(\log r)$ independence queries to $\cM_2$, we can find a vertex $u'\in X$ such that there is an edge $(v, u')$ (i.e., $u'$ is an out-neighbor of $v$).
\end{claim}

We will assume $\textsc{InEdge}(v,X)$
respectively $\textsc{OutEdge}(v,X)$ are procedures which implement \cref{clm:bin-search}.


\section{Algorithms for augmentation} \label{sec:augmentation}

From Claim \ref{clm:exg-graph}, we know the following: Given a common independent set $S$, either $S$ is of maximum cardinality or there exists a (directed) $(s,t)$-path in the exchange graph $G(S)$. In this section, we consider the $(s,t)$-reachability
problem in $G(S)$ using independence oracles.
Our main results in this section are the following two theorems. We denote the size of $S$ as $|S| = r$ in both of these theorems.\footnote{Note that $r$ usually denotes the size of the maximum common independent set which is an upper bound on the size of the vertex set $S$. We abuse the notation and use $r$ here to denote $|S|$.}

\begin{theorem}[Randomized augmentation]
\label{thm:aug}
There is a randomized algorithm which with high probability
uses $O(n\sqrt{r}\log n)$ independence queries and
either determines that $S$ is of maximum cardinality or finds an augmenting path in $G(S)$.
\end{theorem}

\begin{theorem}[Deterministic augmentation]
\label{thm:aug-det}
There is a deterministic algorithm which uses $O(nr^{2/3}\log r)$ independence queries and
either determines that $S$ is of maximum cardinality or finds an augmenting path in $G(S)$.
\end{theorem}

\subsection{Overview of the algorithms} \label{sec:overview}

\Cref{sec:intro:aug-overview} gives an informal overview of the augmentation algorithm already. In this section, we provide more details so that the reader can be convinced about the correctness of the algorithm.

The algorithm for augmentation, denoted as \textsc{Augmentation} algorithm for easy reference, runs in phases and keeps track of a set $F$ of vertices that are reachable from the vertex $s$. Let $F_S$ and $F_{\bar S}$ denote the bipartition of $F$ inside $S$ and $\bar S$, i.e., $F_S = F \cap S$ and $F_{\bar S} = F \cap \bar S$. In each phase, the algorithm will increase the size of $F_S$ by an additive factor of at least $h$ until the algorithm discovers an $(s,t)$-path (or, otherwise, discover there is no such path). Hence, in total, there are at most $\frac{|S|}{h}$ many phases. We now give an overview of how to implement each phase.

Note that, without loss of generality, we can assume that the set $F_S$ contains all vertices that are out-neighbors of vertices in $F_{\bar S}$. This is because whenever a vertex $v\in \bar{S}$ is added to $F_{\bar S}$, we can quickly add all of $v$'s out-neighbors in $S\setminus F_S$ into the set $F_S$ by using Claim \ref{clm:bin-search}. This requires $O(\log r)$ independence queries for each such out-neighbor. Hence, in total, this procedure uses at most $O(n\log r)$ independence queries, since
each $u\in S \setminus F_S$ is added in $F_S$ at most once.

\paragraph{Heavy and light vertices.} Before explaining what the algorithm does in each phase, we introduce the notion of \emph{heavy} and \emph{light} vertices: We divide the vertices in $\bar{S}\setminus F_{\bar S}$ into two categories. We call a vertex $v\in \bar{S}\setminus F_{\bar S}$ \emph{heavy} if it either has an edge to $t$ or has at least $h$ out-neighbors in $S\setminus F_S$. The vertices in $\bar{S}\setminus F_{\bar S}$ that are not \textit{heavy} are denoted as \emph{light} (See \Cref{fig:rev-bfs} for reference; the heavy nodes are highlighted in light-yellow). Note that both these notions are defined in terms of out-degrees, i.e., a heavy vertex can have arbitrary in-degree and so can a light vertex. Also, note that the notion of heavy and light vertices are defined w.r.t. to the set $F_S$. Because the set $F_S$ changes from one phase to the next, so does the set of heavy vertices and light vertices.

\paragraph{Description of phase $i$.} Let us assume, for the time being, that there is an efficient procedure to categorize the vertices in $\bar S \setminus F_{\bar S}$ into the sets of heavy and light vertices. We first apply this procedure at the beginning of phase $i$.

Now, for simplicity, consider an easy case: In phase $i$, there is a \textit{heavy} vertex that has an in-neighbor in $F_S$. In this case, we can go over all vertices in $\bar S \setminus F_{\bar S}$ to find such a heavy vertex---this can be done with $n$ many independence queries. Once we find such a heavy vertex, we include it in $F_S$ and all of its out-neighbors in $F_{\bar S}$. Note that, in this case, either of the following two things can happen: either we have increased the size of $F_S$ by at least $h$ as the heavy vertex has at least $h$ out-neighbors in $S \setminus F_S$; or the heavy vertex we found has $t$ as its out-neighbor in which case we have found an $(s,t)$-path.

Unfortunately, this may not be the case in phase $i$. In this case, we do an additional procedure called the \textit{reverse breadth-first search} or, in short, \textit{reverse BFS}. The goal of the reverse BFS is to find a heavy vertex reachable from $F$.
Before describing this procedure, note the following two properties of the \textit{light} vertices: \begin{enumerate}
    \item \label{itm:pro-2-light} A light vertex will remain a light vertex even if we increase the size of $F_S$.
    
    \item \label{itm:pro-1-light} We can assume that we know all out-neighbors of any light vertex. 
\end{enumerate}
Property \ref{itm:pro-1-light} needs some explanation. This property is true because of two observations: (i) All out-neighbors of a light vertex can be found out with $O(h \log n)$ independence queries using Claim \ref{clm:bin-search}, and (ii) because of Property \ref{itm:pro-2-light}, across all phases, we need to find out the out-neighbors of a light vertex \textit{only once}. So, even though we need to make $O(n h \log n)$ queries in total, this cost amortizes across all phases.

The idea is, as before, to discover a heavy vertex which is reachable from $F$ so that we can include all of its out-neighbors in $F_S$ (for example, consider the heavy vertex $v_1$ in \Cref{fig:rev-bfs}).
So our goal is to find some path from $F$ to a heavy vertex (Consider the path starting from $v_1$ highlighted in light-blue in \Cref{fig:rev-bfs}). This naturally implies the need for doing a reverse BFS from the heavy vertices.
We also note that any path from $F$ to $t$ must pass through a heavy vertex
(the vertex just preceding $t$ must by definition be heavy).
Hence, if our reverse BFS fails to find a path from $F$ to some heavy vertex,
the algorithm has determined that no $(s,t)$-path exists.

What remains is to find out how to implement the reverse BFS procedure efficiently.
To this end, we exploit Property \ref{itm:pro-1-light} of light vertices and assume that we know all edges directed from $\bar S$ to $S$ that the reverse BFS procedure needs to visit. This follows from the following crucial observation: \textit{No internal node of the reverse BFS forest is a heavy node}, i.e., in other words, the heavy vertices occur \textit{only} as root nodes of the reverse BFS trees. This is because if, along the traversal of a reverse BFS procedure starting from a heavy node $v$, we reach another heavy node $v'$, we can ignore $v'$ as the reverse BFS starting from node $v'$ has already taken care of processing $v'$.
This means that any edge in $\bar S \times S$ that takes part in the reverse BFS procedure must originate from a light vertex and, hence, is known a priori due to Property \ref{itm:pro-1-light}. All it remains for the reverse BFS procedure is to discover in-neighbors of vertices in $\bar S$ using edges from $S \times \bar S$. By Claim \ref{clm:bin-search}, each such in-neighbor can be found by making $O(\log r)$ independence queries. In total, the reverse BFS procedure uses $\tilde O(n)$ independence queries. 

\paragraph{Post-processing.} Note that, in order to use Claim \ref{clm:exg-graph}, the $(s,t)$-path needs to be chordless. However, the $(s,t)$-path $p$ that the algorithm outputs has no such guarantee. So, as a post-processing step, the algorithm uses an additional $\tO(r)$ independence queries to convert this path into a \emph{chordless} path: Consider any vertex $v \in V(p) \cap \bar S$ and assume $u$ as the parent of $v$, and $w$ as the child of $v$ in the path $p$. The vertex $v$ needs to check whether it has an in-neighbor other than $u$ among the ancestors of $v$ in $V(p)$ or an out-neighbor other than $w$ among the descendants of $v$ in $V(p)$. Since the length of the path obtained from the previous step is $O(r)$ (because of $|S| = r$ and the path does not contain any cycle), this requires $O(\log r)$ independence queries. If all vertices in $V(p) \cap \bar S$ have no such in or out-neighbors, then it is easy to see that $p$ is indeed a chordless path. If there is such a (say) in-neighbor $u'$ of $v$, then we remove all vertices of $V(p)$ between $u'$ and $v$, and the resulting subsequence is still an $(s,t)$-path. A similar procedure is done when an out-neighbor is discovered. In total, this takes $O(r\log r)$ independence queries, since each vertex can be removed from the path at most once.

\paragraph{Cost analysis.} The total number of queries needed to implement phase $i$ is a summation of two terms: (i) the number of queries needed to partition the vertices into heavy and light categories, and (ii) the number of queries needed to run the reverse BFS procedure. We have seen that (ii) can be implemented with $\tilde O(n)$ independence queries. For (i), we present two algorithms: a randomized sampling algorithm, and a deterministic algorithm which is slightly less efficient than the randomized one. This is the main technical difference between the algorithm of Theorem \ref{thm:aug} and that of Theorem \ref{thm:aug-det}. The cost analysis for (i) is also amortized and the total number of queries needed across all phases is $\tO(\max\{n h, nr/h\} )$ for randomized and $\tO(n \sqrt{rh})$ for deterministic implementation. Setting $h= \sqrt r$ for randomized and $h = r^{1/3}$ for deterministic, we see that total randomized query complexity of augmentation is $\tO(n \sqrt r)$ and deterministic query complexity is $\tO(n r^{2/3})$. \jan{ I assume we can set $h$ like this because we know an approximation of $r$ by running that other approximate algorihtm first?}

\subsection{Categorizing \emph{heavy} and \emph{light} vertices}\label{sec:heavy-light-small}

We start with reminding the readers the definition of the \textit{heavy} and \textit{light} vertices in $\bar S \setminus F_{\bar S}$.
\begin{definition}\label{def:heavy-light}
We call a vertex $v\in \bar{S}\setminus F_{\bar{S}}$ {heavy} if either
$(v,t)$ is an edge of $G(S)$ or $v$ has at least
$h$ out-neighbors in $S\setminus F_{S}$.
Otherwise we call $v$ {light}.
\end{definition}

To check whether $v$ has an edge to $t$ is easy and requires
only a single independence query: ``Is $S \cup \{v\}$ independent in $\cM_2$?''
The difficulty
lies when this is not the case and we need to determine if
$v$ has outdegree at least $h$ to $S\setminus F_{S}$.
We present two algorithms to solve this categorization problem:
one randomized sampling algorithm; and
a less efficient deterministic algorithm. More concretely, we show the following two lemmas.

\begin{restatable}{lemma}{RandCat}\label{lem:sample}
There is a randomized categorization procedure which, with high probability,
categorizes heavy and light vertices in the set $\bar S \setminus F_{\bar S}$ correctly by issuing $O(n\log n)$ independence queries per phase and an additional $O(nh\log n)$ independence queries over the whole run of the \textsc{Augmentation} algorithm. 
\end{restatable}

\begin{restatable}{lemma}{DetCat}\label{lem:det-cat}
There exists a deterministic categorization procedure which uses
$O(n\sqrt{rh}\log r)$ queries over the whole run of the \textsc{Augmentation} algorithm.
\end{restatable}

\noindent
The proofs of these two lemmas are deferred to Section \ref{sec:heavy-light}.

\subsection{Heavy vertex reachability} \label{sec:rev-bfs}
In this section, we present the reverse BFS in \cref{alg:bfs} and
analyze some properties of it. Recall that the reverse BFS
is run once in each phase of the algorithm to find some vertex in
$F$ which can reach some heavy vertex.
We also remind the reader of the example in \Cref{fig:rev-bfs}.
In this section, we prove the following.

\begin{lemma}[Heavy vertex reachability] \label{lem:bfs}
There is an algorithm (\cref{alg:bfs}: \textsc{ReverseBFS}) which, given $F$ such that there are no edges from $F_{\bar{S}}$ to $S\setminus F_{S}$,
a categorization of $\bar{S}\setminus F_{\bar{S}}$ into \emph{heavy}
and \emph{light}, and all out-edges of the light vertices to $S\setminus F_{S}$, uses $O(n\log r)$ queries and either finds a path
from some vertex in $F$ to a heavy node, or otherwise
determines that no such path exists.
\end{lemma}

\noindent
We next provide the pseudo-code (Algorithm \ref{alg:bfs}). 

\begin{center}
  \centering
  \begin{minipage}[H]{0.8\textwidth}
\begin{algorithm}[H]
\caption{\textsc{ReverseBFS}}\label{alg:bfs}
\begin{algorithmic}[1]
\Statex \textbf{Input:} Categorization of $\bar{S}\setminus F_{\bar{S}}$ into
\emph{heavy} and \emph{light}; and a set
\LightEdges{} containing all out-edges of the light vertices.
\Statex \textbf{Output:} A path from $F$ to some heavy vertex, if one exists.
\Statex \hrulefill
\State $Q \gets \{v\in \bar{S}\setminus F_{\bar{S}} \text{ which are heavy}\}$
\State $\NotVisited \gets (S\cup \bar{S}\cup\{s,t\})\setminus Q$
\While{$Q \neq \emptyset$}
    \State Pop a vertex $v$ from $Q$.
    \If{$v\in F$}
        \State \Return the path from $v$ to a heavy vertex in the BFS-forest.
    \ElsIf{$v\in \bar{S}\setminus F_{\bar{S}}$}
        \While{$u = \textsc{InEdge}(v,\NotVisited)$ is not $\emptyset$}
        \label{lst:line:bfs-inedge}
            \State Push $u$ to $Q$ and remove it from $\NotVisited$.
        \EndWhile
    \ElsIf{$v\in S\setminus F_S$}
        \For{$u\in \NotVisited$ such that $(u,v)\in \LightEdges$}
            \State Push $u$ to $Q$ and remove it from $\NotVisited$.
        \EndFor
    \EndIf
\EndWhile
\State \Return ``NO PATH EXISTS''
\end{algorithmic}
\end{algorithm}
\end{minipage}
\end{center}

\paragraph{Correctness.}
We first argue that the algorithm is correct.
When a vertex $v\in \bar{S}\setminus F_{\bar{S}}$ is processed by the algorithm,
each unvisited in-neighbor will be added to the queue $Q$ in the while loop
in line~\ref{lst:line:bfs-inedge}.
When a vertex $v\in S\setminus F_{S}$ is processed by the algorithm,
any edge from $\NotVisited$ to $v$ must originate from a light vertex,
since $\NotVisited$ contains no heavy vertices and we are guaranteed
that no edge from $F_{\bar{S}}$ to $S\setminus F_S$ exist.
Hence \cref{alg:bfs} will eventually process every vertex
reachable, by traversing edges in reverse, from the heavy vertices.

\paragraph{Cost analysis.}
The only place \cref{alg:bfs} uses independence queries is in line~\ref{lst:line:bfs-inedge}.
Each vertex will be discovered at most once by the binary search in $\textsc{InEdge}$.
This means that we do at most $n$ calls to $\textsc{InEdge}$, each using $O(\log r)$ queries by \cref{clm:bin-search}. Hence the reverse BFS uses $O(n\log r)$ queries per
phase.

\subsection{Augmenting path algorithm.}

We now present the main augmenting path algorithm, as explained in
the overview in \cref{sec:overview}.

\begin{center}
  \centering
  \begin{minipage}[H]{0.8\textwidth}
\begin{algorithm}[H]
\caption{\textsc{Augmentation}}\label{alg:augmentation}
\begin{algorithmic}[1]
\Statex \textbf{Input:} Two matroids $\cM_1 = (V,\cI_1)$,
and $\cM_2 = (V,\cI_2)$ and a common independent set $S\subseteq \cI_1 \cap \cI_2$.
\Statex \textbf{Output:} An augmenting $(s,t)$-path in $G(S)$ if one exists.
\Statex \hrulefill

\State $F\gets \{s\}$
\State $\LightEdges \gets \emptyset$

\Statex

\While{$t \not\in F$}
\begin{mdframed}[backgroundcolor=blue!20,skipabove=5pt,skipbelow=5pt]
      \centering
      {\em Description of a phase.}
    \end{mdframed}
    \State Categorize $v\in \bar{S}\setminus F_{\bar{S}}$ into \emph{heavy}
    and \emph{light}. \Comment{See \cref{sec:rand,sec:det}.}
    \For{each new \emph{light} vertex $v$} 
    \State Use \textsc{OutEdge} to find all out-neighbors of $v$ in $S\setminus F_S$.
    \label{lst:line:light-edge}
    \State Add edges $(v,u)$ to $\LightEdges$ for each such out-neighbor $u$.
    \EndFor
    \Statex
    \State $p \gets \textsc{ReverseBFS}(S,F,\LightEdges)$ \Comment{See Section \ref{sec:rev-bfs}.}
    \If{$p = $ ``NO PATH EXISTS''}
        \State \Return ``NO PATH EXISTS''
    \Else \Comment{$p$ is a path from $F$ to a heavy vertex}
        \State Denote by $V(p)$ the vertices on the path $p$.
        \State Add all $v\in V(p)$ to $F$.
        \For{$v \in V(p)\cap \bar{S}$}
            \While{$u = \textsc{OutEdge}(v,(S\setminus F)\cup\{t\})$ is not $\emptyset$}
                \label{lst:line:uw-edge}
                \State Add $u$ to $F$.
            \EndWhile
        \EndFor
    \EndIf
\EndWhile

\begin{mdframed}[backgroundcolor=blue!20,skipabove=5pt,skipbelow=5pt]
      \centering
      {\em Post-processing.}
    \end{mdframed}
\State Post-process the $(s,t)$-path found to make it \emph{chordless}.
\State \Return the augmenting path.

\end{algorithmic}
\end{algorithm}
\end{minipage}
\end{center}

Note that we have not specified if we are using
the randomized or
deterministic categorization of heavy and light vertices,
from \cref{sec:rand,sec:det}.
We will for now assume this categorization procedure as a black box which
is always correct. 

We start by stating some invariants of \cref{alg:augmentation}.
\begin{enumerate}
\item \label{itm:inv-F} $F$ contains only vertices reachable from $s$. In fact, for each vertex in $F$
we have found a path from $s$ to this vertex.
\item $\LightEdges$ contains all out-edges from light vertices to $S\setminus F$.
\item \label{itm:inv-Fclosure} In the beginning of each phase, there exists no $v\in F$, $u\in (S\cup\{s,t\})\setminus F$
      such that $(v,u)$ is an edge in $G(S)$. This is because whenever
      $v\in \bar{S}$ is added to $F$, all $v$'s neighbors are also added,
      see line~\ref{lst:line:uw-edge}.
\end{enumerate}

\paragraph{Correctness.}
When the algorithm outputs an $(s,t)$-path, the path clearly exists,
by Invariant~\ref{itm:inv-F}.
So it suffices to argue that the algorithm does not return ``NO PATH EXISTS'' incorrectly.
Note that the algorithm only returns ``NO PATH EXISTS'' when
\textsc{ReverseBFS} does so, that is when there is no
path from $F$ to a heavy vertex (by \cref{lem:bfs}).
So suppose that this is that case, and also suppose, for the sake of a contradiction,
that an $(s,t)$-path $p$ exists in $G(S)$.
Denote by $v$ the vertex preceding $t$ in the path $p$.
By Invariant~\ref{itm:inv-Fclosure} we know that $v$ is not in $F$.
But then $v$ is heavy, since $(v,t)$ is an edge of $G(S)$.
Hence a subpath of $p$ will be a path from $F$ to the heavy vertex $v$,
which is the desired contradiction.


\paragraph{Number of phases.}
We argue that there are at most $\frac{r}{h}+1$ phases of the algorithm.
After a phase, either the algorithm returns ``NO PATH EXISTS'' (in which case
this was the last phase), or some path $p$ was found
by the reverse BFS. Then $V(p)$ must include some heavy vertex $v$.
Then all neighbors of $v$ will be added to $F$ in line~\ref{lst:line:uw-edge}.
Thus we know that either $t$ was added to $F$ (in which case this was
the last phase), or at least $h$ vertices from $S$ was added to $F$.
Since $|S| \le r$ in the beginning, this can happen at most
$\frac{r}{h}$ times.

\paragraph{Number of queries.}
We analyse the number of independence queries used by different parts of the algorithm:
\begin{itemize}
    \item \textsc{ReverseBFS} (\cref{alg:bfs})
    is run once each phase, and uses $O(n\log r)$
    queries per call by \cref{lem:bfs}.
    This contributes a total of $O(\frac{nr\log r}{h})$ independence queries
    over all phases.
    
    \item Each $u\in S$ is discovered at most once by the
    \textsc{OutEdge} call on line~\ref{lst:line:uw-edge}.
    So this line contributes a total of $O(n\log r)$ independence queries.
    
    \item Each vertex becomes \emph{light} at most once over the run of the algorithm.
    When this happens, the algorithm 
    finds all of its (up to $h$) out-neighbors on line~\ref{lst:line:light-edge},
    using \textsc{OutEdge} calls.
    This contributes a total of $O(nh\log r)$ independence queries.
    
    \item The post-processing  can be performed using $O(r\log r)$ independence
    queries, as explained in \cref{sec:overview}.
    
    \item The \emph{heavy}/\emph{light}-categorization
    uses $O(nh\log n + \frac{nr\log n}{h})$ independence queries when the randomized
    procedure is used, by \cref{lem:sample}.
    When the deterministic categorization procedure is used, 
    we use $O(n\sqrt{rh}\log r)$ independence queries instead,
    by \cref{lem:det-cat}.
\end{itemize}
We see that in total, the algorithm uses:
\begin{itemize}
    \item $O(n\sqrt{r}\log n)$ independence queries with the randomized
    categorization, setting $h = r^{1/3}$.
    \item $O(nr^{2/3}\log r)$ independence queries with the deterministic
    categorization, setting $h = \sqrt{r}$.
\end{itemize}
The above analysis proves \cref{thm:aug,thm:aug-det}.

\begin{remark}
When the randomized categorization procedure fails, \cref{alg:augmentation} will
still always return the correct answer, but it might use more independence queries.
So \cref{alg:augmentation} is in fact a Las-Vegas algorithm with expected query-complexity
$O(n\sqrt{r}\log n)$.
\end{remark}
\begin{remark}
\label{rem:reachable}
We note that our algorithm can not be used to find which vertices are reachable from $s$ using subquadratic number of queries.
\end{remark}

\section{Algorithm for fast Matroid Intersection} \label{sec:main-algo}

There are two hurdles to getting a subquadratic algorithm for Matroid Intersection.
Firstly, standard augmenting path algorithms need to find the augmenting paths one at
a time. This is since after augmenting along a path, the edges in the exchange graph
change (some edges are added, some removed).
This is unlike bipartite matching, where a set of vertex-disjoint augmenting
paths can be augmented along in parallel. It is not clear how to find
the augmenting paths faster than $\Theta(n)$ each, so these standard
augmenting path algorithms are stuck at $\Omega(nr)$ independence queries.

To overcome this,  Chakrabarty-Lee-Sidford-Singla-Wong \cite[Section~6]{chakrabarty2019faster} introduce the notion of
\emph{augmenting sets}, which allows multiple parallel augmentations.
Using the augmenting sets they present a
subquadratic $(1-\eps)$-approximation algorithm using
$\tO(\frac{n^{1.5}}{\eps^{1.5}})$ independence queries:

\begin{lemma}[Approximation algorithm \cite{chakrabarty2019faster}]
\label{lem:approx}
There exists an $(1-\eps)$ approximation algorithm for matroid intersection
using $O(\frac{n\sqrt{n\log r}}{\eps\sqrt{\eps}})$ independence queries.
\end{lemma}

The second hurdle is that when the distance $d$ between $s$ and $t$ is high, the
breadth-first algorithms of \cite{chakrabarty2019faster,nguyen2019note} use $\tilde{\Theta}(dn)$ independence queries
to compute the distance layers, which is $\Omega(nr)$ when $d \approx r$.\footnote{Note that unlike in \cref{sec:augmentation}, we now use the normal definition of $r$ as the size of the maximum-cardinality common independent set of the two matroids.}
Here our algorithm from \cref{sec:augmentation} helps since it can find a single
augmenting path using a subquadratic number of independence queries, even when the distance
$d$ is large.

So our idea is as follows:
\begin{itemize}
    \item Start by using the subquadratic approximation algorithm. This avoids having to do the majority of augmentations one by one.
    \item Continue with the fast implementation \cite[Section~5]{chakrabarty2019faster} of the Cunningham-style blocking flow algorithm.
    \item When the $(s,t)$-distance becomes too large, fall back to using the augmenting-path algorithm from \cref{sec:augmentation} to find the (few) remaining augmenting paths.
\end{itemize}

\begin{center}
  \centering
  \begin{minipage}[H]{0.8\textwidth}
\begin{algorithm}[H]
\caption{subquadratic Matroid Intersection}\label{alg:subquadratic}
\begin{algorithmic}[1]
\State
       Run the approximation algorithm (\cref{lem:approx})
       with $\eps = n^{1/5}r^{-2/5}\log^{-1/5} r$
       to obtain a common independent set $S$ of size at least
       $(1-\eps)r = r-n^{1/5}r^{3/5}\log^{-1/5} r$.\label{lst:line:approx}
\Statex
\State Starting with $S$, run Cunningham's algorithm (as implemented by
       \cite{chakrabarty2019faster}), until the distance between $s$ and $t$
       becomes larger than $d$.
       \label{lst:line:cunningham}
\Statex
\State Keep running \textsc{Augmentation} (\cref{alg:augmentation}) from \cref{sec:augmentation} and augmenting the current common independent set with the obtained $(s,t)$-path (as in Claim \ref{clm:exg-graph}) until no $(s,t)$-path can be found in the exchange graph.
       \label{lst:line:augmenting-paths}
\end{algorithmic}
\end{algorithm}
\end{minipage}
\end{center}

The choice of $d$ will be different depending on whether we use
the randomized or deterministic version of \cref{alg:augmentation}.
In order to run \cref{alg:subquadratic},
we need to know $r$ so that we may choose $\eps$ (and $d$) appropriately.
However, the size $r$ of the largest common independent set is unknown.
We note that it suffices, for the purpose of the asymptotic analysis, to use a $\frac{1}{2}$-approximation $\bar{r}$
for $r$ (that is $\bar{r} \le r \le 2\bar{r}$). It is well known
that such an $\bar{r}$ can be found
in $O(n)$ independence queries by greedily finding a maximal
common independent set in the two matroids.
Now we can bound the query complexity of  \cref{alg:subquadratic}.

\begin{lemma} \label{lem:line-1}
Line~\ref{lst:line:approx} of \cref{alg:subquadratic} uses $O(n^{6/5}r^{3/5}\log^{4/5} r)$ independence queries.
\end{lemma}
\begin{proof}
The approximation algorithm uses $O(\frac{n^{1.5}\sqrt{\log r}}{\eps^{1.5}}) = O(n^{6/5}r^{3/5}\log^{4/5} r)$ independence queries, when $\eps = n^{1/5} r^{-2/5}\log^{-1/5} r$.
\end{proof}

\begin{lemma} \label{lem:line-2}
Line~\ref{lst:line:cunningham} of \cref{alg:subquadratic} uses $O(n^{6/5}r^{3/5}\log^{4/5}r + nd\log r)$ independence queries.
\end{lemma}
\begin{proof}
There are two main parts of Cunningham's blocking-flow algorithm.
\begin{itemize}
    \item Computing the distances.
    The algorithm will run several BFS's to compute the distances.
    The total number of independence queries for all of these BFS's can be bounded
    by $O(d n \log r)$, since the distances are monotonic so each vertex is
    tried at a specific distance at most once.
    For more details, see \cite[Section~5.1]{chakrabarty2019faster}.
    
    \item Finding the augmenting paths.
    Given the distance-layers, a single augmenting path can be found in $O(n\log r)$
    independence queries, by a simple depth-first-search. Again,
    we refer to \cite[Section~5.2]{chakrabarty2019faster} for more details.
    Since we start with a common independent set $S$
    of size $(1-\eps)r = r-n^{1/5}r^{3/5}\log^{-1/5} r$,
    we know that $S$ can be augmented at most $n^{1/5}r^{3/5}\log^{-1/5}r$ additional times. Hence a total of
    $O(n^{6/5}r^{3/5}\log^{4/5} r)$ independence queries suffices to find all of these augmenting paths.
\end{itemize}
\end{proof}

\begin{remark}
We note that if we skip Line~\ref{lst:line:augmenting-paths} in \cref{alg:subquadratic},
we thus get a $(1-\frac{1}{d})$-approximation algorithm
(by \cref{clm:dists}),
using $\tO(n^{6/5}r^{3/5} + nd)$ independence queries,
which beats the $\tO(n^{1.5}/\eps^{1.5})$ approximation algorithm
when $\eps = o(n^{1/5}r^{-2/5})$.
\end{remark}

\begin{lemma} \label{lem:line-3}
Line~\ref{lst:line:augmenting-paths} of \cref{alg:subquadratic} uses $O(\frac{r}{d}\mathcal{T})$ independence queries, where $\mathcal{T}$ is the number of independence queries
used by one invocation of \textsc{Augmentation} (\cref{alg:augmentation}).
\end{lemma}
\begin{proof}
After line~\ref{lst:line:cunningham}, the algorithm has found a common independent set of size at least $(1-O(\frac{1}{d}))r = r - O(\frac{r}{d})$, by \cref{clm:dists}.
This means that only $O(\frac{r}{d})$ additional augmentations need to be performed.
\end{proof}

By \cref{lem:line-1,lem:line-2,lem:line-3}, we see that \cref{alg:subquadratic}
uses a total of
$O(n^{6/5}r^{3/5}\log^{4/5} r + nd\log r + \frac{r}{d}\mathcal{T})$ independence queries.
If we pick $d = \sqrt{\frac{r\mathcal{T}}{n\log r}}$ we get the following lemma.

\begin{lemma} \label{thm:mi-given-aug}
  If the query complexity of \textsc{Augmentation} is $\cal T$, then matroid intersection can be solved using $O(n^{6/5}r^{3/5}\log^{4/5} r + \sqrt{nr\mathcal{T}\log r})$
  independence queries.\danupon{TO DO: Say something about deterministic/randomized like intro.}
\end{lemma}

Combining with \cref{thm:aug-det,thm:aug} we get our subquadratic results.

\begin{theorem} [Randomized Matroid Intersection]
\label{thm:mi-rand}
There is a randomized algorithm which with high probability
uses $O(n^{6/5}r^{3/5}\log^{4/5} r)$ independence queries and
solves the matroid intersection problem.
When $r = \Theta(n)$, this is $\tO(n^{9/5})$.
\end{theorem}

\begin{theorem}[Deterministic Matroid Intersection]
\label{thm:mi-det}
There is a deterministic algorithm which
uses $O(nr^{5/6} \log r + n^{6/5}r^{3/5}\log^{4/5} r)$ independence queries and
solves the matroid intersection problem.
When $r = \Theta(n)$, this is $\tO(n^{11/6})$.
\end{theorem}

\begin{remark}
The limiting term for the the randomized algorithm is between line~\ref{lst:line:approx} and line~\ref{lst:line:cunningham}.
If a faster approximation algorithm is found, the same strategy as above might give
an $\tO(nr^{3/4})$-query algorithm.
\end{remark}

\section{Algorithm for heavy/light categorization} \label{sec:heavy-light}

In this section, we finally provide the algorithm for the categorization of vertices in $\bar S \setminus F_{\bar S}$ into heavy and light vertices as defined in Definition \ref{def:heavy-light}.

\subsection{Randomized categorization} \label{sec:rand}

In this section, we prove the following lemma (restated from Section \ref{sec:heavy-light-small}).

\RandCat*

We will use $X$ to denote $S\setminus F$. Let the out-neighborhood of a vertex $v \in \bar S \setminus F_{\bar S}$ inside $X$ be denoted as $\ngh_X(v)$. Consider the family of sets $\{\ngh_X(v)\}_{v \in \bar S \setminus F_{\bar S}}$ residing inside the ambient universe $X$. We want to find out which of these sets are of size at least $h$ (i.e., correspond to the heavy vertices) and which of them are not (i.e., corresponds to the light vertices). To this end, we devise the following random experiment.

\begin{experiment} \label{exp}
Sample a set $R$ of $k$ elements drawn uniformly and independently from $X$ (with replacement) and check whether $R \cap \ngh_X(v) = \emptyset$. 
\end{experiment}

It is easy to check the following: For any $v \in \bar S \setminus F_{\bar S}$, \cref{exp} is successful with probability:
\[
\Pr_R[R \cap \ngh_X(v) = \emptyset] = \left(1 - \frac{|\ngh_X(v)|}{|X|}\right)^k.
\]

Note that, to perform this experiment for a vertex $v$, we need to make a single independence query of the form whether $(S \setminus R) \cup \{v\} \in \cI_2$. Next, we make the following claim.

\begin{claim}\label{clm:exp-succ}
There is a non-negative integer $k$ such that the following holds: \begin{enumerate}
    \item  If $|\ngh_X(v)| < h$, then \cref{exp} succeeds with probability at least 3/4, and
    
    \item \label{itm:exp-case-2} If $|\ngh_X(v)| > 10h$, then \cref{exp} succeeds with probability at most 1/4.
\end{enumerate}
\end{claim}

Before proving Claim \ref{clm:exp-succ}, we show the rest of the steps of this procedure. For every vertex, we repeat \cref{exp} $s = O(\log n)$ many times independently. By standard concentration bound, we make the following observations: \begin{enumerate}
    \item \label{itm:exp-1} If $|\ngh_X(v)| < h$, strictly more than $s/2$ experiments succeed with very high probability.\footnote{Recall that by \emph{very high probability} we mean with probability at least $1 - n^{-c}$ for some arbitrary large constant $c$.}
    
    \item \label{itm:exp-2} If $|\ngh_X(v)| > 10h$, strictly less than $s/2$ experiments succeed with very high probability.
\end{enumerate}

Hence, we declare any vertex for which strictly less than $s/2$ experiments succeed as \textit{heavy}. The probability that a light vertex can be classified as heavy by this procedure is very small due to Property \ref{itm:exp-1}. On the other hand, a vertex with $|\ngh_X(v)| > 10h$ will be correctly classified as heavy with a very high probability. However, a heavy vertex with $|\ngh_X(v)| \leq 10h$ may not be correctly classified. So, for such vertices, we want to check in a brute-force manner. To this end, we discover the set $\ngh_X(v)$ for any vertex $v$ which is not declared heavy and make decisions accordingly. 

\paragraph{Bounding the error probability.} 
We argue that we can bound the error probabilities from
Properties~\ref{itm:exp-1} and~\ref{itm:exp-2}
over the whole run of the \textsc{Augmentation} algorithm by a union bound.
Say that 
the error probabilities of Properties~\ref{itm:exp-1} and~\ref{itm:exp-2}
is bounded by $n^{-c}$ for some large constant $c \ge 10$.
In each phase we categorize at most $n$ vertices,
and there is at most $\frac{r}{h} < n$ phases.
Hence, the probability that --- over the whole run of
\textsc{Augmentation} (\cref{alg:augmentation}) ---
that any vertex is misclassified as heavy,
or that the procedure decides to discover a set $\ngh_X(v)$ with
$\ngh_X(v)> 10$, is at most $n^{-c+2}$.
Similarly we note that in the algorithm for Matroid Intersection
(\cref{alg:subquadratic}) we run \textsc{Augmentation} at most
$r$ times, so the error probability is at most $n^{-c+3}$.

\paragraph{Cost analysis.} As mentioned before, each instance of \cref{exp} can be performed with a single query. As there are $O(n \log n)$ experiments in total in each phase of the algorithm, the number of queries needed to perform all
experiments 
over the whole run of the \textsc{Augmentation} algorithm will be
is $O(\frac{nr\log n}{h})$ (recall that the number of phases is $r/h$). Now consider the part of the algorithm where we need to discover the set $\ngh_X(v)$ for any vertex $v$ which is not declared heavy after the completion of all experiments in a phase. For each such vertex, this will take at most $O(|\ngh_X(v)|\log n) = O(h\log n)$ queries (due to Claim \ref{clm:bin-search}). Note that we only need to make these kinds of queries from each vertex once over the whole run of the algorithm (as in future queries we already know all $v$'s neighbors and can answer directly). Hence, the total number of such queries is at most $O(nh\log n)$ across all phases of the algorithm.

\begin{proof}[Proof of Claim \ref{clm:exp-succ}]
First we note that if $|X| \le 10h$, case~\ref{itm:exp-case-2} is vacuously true, so we may pick $k = 0$ such that
\cref{exp} always succeeds. So now assume that $|X| > 10h$
and let $x = \frac{h}{|X|}\in (0,\frac{1}{10})$. We want to show
that there exists some positive integer $k$ satisfying
$(1-x)^k \ge \frac{3}{4}$ and
$(1-10x)^{k} \le \frac{1}{4}$.
Pick $k = \left\lceil\frac{\log{\frac{1}{4}}}{\log(1-10x)} \right\rceil$.
Then $k \ge \frac{\log{\frac{1}{4}}}{\log(1-10x)} > 0$,
which means that $(1-10x)^k\le \frac{1}{4}$.
We also have that
$k \le \frac{\log{\frac{1}{4}}}{\log(1-10x)} + 1
< \frac{\log {\frac{3}{4}}}{\log(1-x)}$ (since $x\in (0,\frac{1}{10})$),
which means that $(1-x)^k\ge \frac{3}{4}$.
\end{proof}

\subsection{Deterministic categorization} \label{sec:det}
In this section we prove the following lemma (restated from Section \ref{sec:heavy-light-small}).

\DetCat*

The main idea of the determinsitic categorization is the following: For each  $v\in \bar{S}\setminus F_{\bar{S}}$, our deterministic
categorization keeps track of a set $N_v\subseteq \ngh_X(v)$ of $h$
out-neighbors to $v$ (if that many out-neighbors exist).
Then we can either use $N_v$ as a proof that $v$ is heavy,
or when we failed to find such a $N_v$ we know that $v$ is light.

In each phase, some of the vertices in $N_v$ may be added to $F$ (and
thus removed from $X$).
This may decrease the size of $N_v$. In this case we would like
to find additional out-neighbors to add to $N_v$, until $|N_v| = h$,
or determine that $|\ngh_X(v)| < h$. One possible and immediate strategy would be to use \cref{clm:bin-search} to find a new out-neighbor of $v$ in $O(\log n)$ independence queries.
However, adding arbitrary neighbors from $\ngh_X(v)\setminus N_v$ will be expensive:
over the whole run of the algorithm potentially every vertex in $S$ will be added to $N_v$ at some point which will require $\tO(nr)$ many independence queries in total for all $N_v$'s---this is far too expensive than what we can allow. Instead, we want to be device a better strategy to pick $u\in \ngh_X(v)\setminus N_v$.

\paragraph{Determinisitc strategy.} For $u\in X$ we will denote by the \emph{weight} of $u$, or $w(u)$,
the number of sets $N_v$ which contain $u$. Note that these weights
change over the run of the algorithm. Also, note that the values $w(u)$ can be inferred from the sets $N_v$'s which are known to the querier. Hence, we can assume that the querier knows the weights of elements in $X$. When $u$ is moved from $X$ to $F$, $w(u)$ new out-neighbors must be found,
one for each $v\in\bar{S}\setminus F_{\bar{S}}$
for which the set $N_v$ contained $u$. 

This motivates the following strategy:
\emph{Whenever we need to find a new out-neighbors of $v$,
we find $u\in \ngh_X(v)\setminus N_v$ that minimizes $w(u)$}.
To perform this strategy, we note that the binary-search
idea from \cref{clm:bin-search} can be implemented to find a
$u$ which minimizes $w(u)$.
Indeed, if $\{u_1, u_2, \ldots u_{|X|}\}\subseteq X$ with
$w(u_1) \le w(u_2) \le \ldots \le w(u_{|X|})$, the binary search
can first ask if there is an edge to $\{u_1, \ldots, u_{\floor{|X|/2}}\}$ with
a single query. If this was the case we recurse on
$\{u_1, u_2, \ldots u_{\floor{|X|/2}}\}$, otherwise  recurse
on  $\{u_{\floor{|X|/2}+1}, \ldots, u_{|X|} \}$. This will guarantee that a
the $u_i$ which minimizes $w(u_i)$ will be found.\footnote{We actually use the same strategy to initialize the sets $N_v$: We discover out-neighbors $u$ in the increasing order of $w(u)$.}

\paragraph{Cost Analysis.}
For each $v\in \bar{S}$ we will at most
once determine that $N_v$ cannot be extended, i.e. that $|\ngh_X(v)| < h$. This will
require $O(n)$ independence queries in total.
The remaining cost we will amortize over the vertices in $V = S\cup \bar{S}$.
Consider that we find some out-neighbor $u\in X$ to some vertex $v\in \bar{S}$, using the above strategy. This uses $O(\log r)$ independence queries.
We will charge this cost to $u$ if $w(u) \le \frac{n\sqrt{h}}{\sqrt{r}}$,
otherwise we will charge the cost to $v$. We make the following observations:
\begin{enumerate}
    \item \label{itm:charge-1}
    For $u\in S$, the total cost we charge to it at most $O(\frac{n\sqrt{h}}{\sqrt{r}}\log r)$.
    
    \item \label{itm:charge-2}
    For $v\in \bar{S}$, the total cost we charge to it is at most $O(\sqrt{rh}\log r)$.
\end{enumerate}
Property~\ref{itm:charge-1} is easy to see, since we charge the cost
$O(\log r)$ to it at most $O(\frac{n\sqrt{h}}{\sqrt{r}})$ times.
To argue that Property~\ref{itm:charge-2} holds, let $u\in S$ be the first
vertex which got added to $N_v$ which had weight $w(u)$ strictly more than $\frac{n\sqrt{h}}{\sqrt{r}}$
(at the moment it was added to $N_v$).
At this point in time, we know that for all remaining $u'\in \ngh_X(v)\setminus N_v$, must have $w(u') \ge w(u) > \frac{n\sqrt{h}}{\sqrt{r}}$. Note that we can bound the total weight
$\sum_{u\in X} w(u) = \sum_{v\in \bar S\setminus F_{\bar{S}}} |N_v| \le nh$
at any point in time. Because of this upper bound, there can be at most $\frac{nh}{n\sqrt{h}/\sqrt{r}} = \sqrt{rh}$ such $u'$. Hence we can charge vertex $v$
at most $\sqrt{rh}$ more times.

Since there are at most $r$ vertices $u\in S$ and
$n$ vertices $v\in \bar{S}$, we conclude that the total cost (over all phases)
for the deterministic categorization is $O(n\sqrt{rh}\log r)$.
This proves \cref{lem:det-cat}.

\section{Open Problems}
A major open problem is to close the big gap between upper and lower bounds for the matroid intersection problem with independent and rank queries. A major step towards this goal is to prove an $n^{1+\Omega(1)}$ lower bound. It will already be extremely interesting to prove  such a bound for deterministic algorithms. It is also interesting to prove a $cn$ lower bound for randomized algorithms for some constant $c>1$ (the existing lower bound \cite{Harvey08} holds only for deterministic algorithms). 


Another major open problem is to understand whether the rank query is more powerful than the independence query. Are the tight bounds the same under both query-models? Two important intermediate steps towards answering this question is to achieve an $\tilde O(n\sqrt{r})$-query exact algorithm and an $\tilde O(n/\poly(\epsilon))$-query  $(1-\epsilon)$-approximation algorithm under independence queries (such bounds have already been achieved under rank queries \cite{chakrabarty2019faster}). We conjecture that the tight bounds are $\tilde O(n\sqrt{n})$ under both queries when $r=\Omega(n)$.

We believe that fully understanding the complexity of the reachability problem will be another major step towards understanding the matroid intersection problem. We conjecture that our $\tilde O(n\sqrt{n})$ bound is tight for $r=\Omega(n)$. 

It is also very interesting to break the quadratic barrier for the weighted case. 
%
This barrier can be broken by a $(1-\epsilon)$-approximation algorithm by combining techniques from \cite{chakrabarty2019faster,ChekuriQ16}\footnote{From a private communication with Kent Quanrud}, but not the exact one. 

Related problems are those for minimizing submodular functions. Proving an $n^{1+\Omega(1)}$ lower bound or subquadratic upper bound for, e.g., finding the minimizer of a submodular function or the non-trivial minimizer of a symmetric submodular function. Many recent studies (e.g. \cite{RubinsteinSW18,GraurPRW20,LeeLSZ20,MukhopadhyayN20}) have led to some non-trivial bounds. However, it is still open whether an $n^{1+\Omega(1)}$ lower bound or an $n^{2-\Omega(1)}$ upper bound exist even in the special cases of computing minimum $st$-cut and hypergraph mincut in the cut query model.  

\section*{Acknowledgment}
This project has received funding from the European Research Council (ERC) under the European
Unions Horizon 2020 research and innovation programme under grant agreement No 715672. Jan van den Brand is partially supported by the Google PhD Fellowship Program. Danupon
Nanongkai and Sagnik Mukhopadhyay are also partially supported by the Swedish Research Council (Reg.~No. 2019-05622). 

\bibliography{biblio}

\end{document}

%% file: preamble.tex
\usepackage{hyperref}
\hypersetup{bookmarks, breaklinks, urlcolor=blue, colorlinks=true, citecolor=green!50!black, linkcolor=blue}
\usepackage[letterpaper, left=1in, right=1in, top=0.9in, bottom=0.9in]{geometry}
\usepackage[utf8]{inputenc}
\usepackage[american]{babel}
\usepackage{etoolbox}

\usepackage[normalem]{ulem}

\usepackage{graphicx}
\graphicspath{{images/}{../images/}}

\usepackage{amsmath, amssymb, cases}
\usepackage{theoremstyles}
\usepackage{enumitem}
\usepackage{mdframed}
\usepackage{bbm}
\usepackage{bm}

\usepackage[ruled]{algorithm}
\usepackage[noend]{algpseudocode}

\usepackage{microtype}

\usepackage{mdframed}
\usepackage{xcolor}

\usepackage{tikz}
\usetikzlibrary{fit, arrows, shapes, matrix, chains, positioning, backgrounds, arrows.meta, tikzmark, decorations.pathreplacing}

\input{commands}

%% file: commands.tex
\usepackage{modletters}

\usepackage{mathtools}

\newcommand\restr[2]{{
  \left.\kern-\nulldelimiterspace 
  #1 
  \vphantom{\big|} 
  \right|_{#2} 
  }}

\newcommand{\eps}{\varepsilon}

\makeatletter
\def\moverlay{\mathpalette\mov@rlay}
\def\mov@rlay#1#2{\leavevmode\vtop{%
   \baselineskip\z@skip \lineskiplimit-\maxdimen
   \ialign{\hfil$\m@th#1##$\hfil\cr#2\crcr}}}
\newcommand{\charfusion}[3][\mathord]{
    #1{\ifx#1\mathop\vphantom{#2}\fi
        \mathpalette\mov@rlay{#2\cr#3}
      }
    \ifx#1\mathop\expandafter\displaylimits\fi}
\makeatother

\usepackage{graphicx}

\newcommand{\poly}{\mathrm{poly}}

\DeclareMathOperator*{\rk}{rk}

\newcommand{\floor}[1]{\lfloor #1 \rfloor}

\newcommand{\size}[1]{\mathrm{size}}

\newcommand{\set}[2][ ]{\{#2 \ifthenelse{\equal{#1}{ }}{ }{~|~#1}\}}

\newcommand{\comment}[1]{}

\newcommand{\seepage}[2][See]{
    \marginnote{
        \scriptsize {#1} p.~\pageref{#2}
    }
}

\newcommand{\reuse}[1]{
	\expandafter\stepcounter{#1_help}
    \expandafter\label{#1_app}
    \csname#1\endcsname*
}

\newcommand{\ngh}{\mathsf{Ngh}}

%% file: abstract.tex

\begin{abstract}

The matroid intersection problem is a fundamental problem that has been extensively studied for half a century. In the classic version of this problem, we are given two matroids $\cM_1 = (V, \cI_1)$ and $\cM_2 = (V, \cI_2)$ on a comment ground set $V$ of $n$ elements, and then we have to find the largest common independent set $S \in \cI_1 \cap \cI_2$ by making  {\em independence oracle queries}  of the form ``Is $S \in \cI_1$?'' or ``Is $S \in \cI_2$?'' for $S \subseteq V$. The goal is to minimize the number of queries. 

Beating the existing $\tO(n^2)$ bound, known as the {\em quadratic barrier}, is an open problem that captures the limits of techniques from two lines of work. The first one is the classic Cunningham's algorithm [SICOMP 1986], whose $\tO(n^2)$-query implementations were shown by CLS+ [FOCS 2019] and Nguy\~{\^e}n [2019].\footnote{More generally, these algorithms take $\tO(nr)$ queries where $r$ denotes the rank which can be as big as $n$.}  The other one is the general cutting plane method of Lee, Sidford, and Wong [FOCS 2015]. The only progress towards breaking the quadratic barrier requires either {\em approximation} algorithms or a more powerful \textit{rank oracle query} [CLS+ FOCS 2019]. No exact algorithm with $o(n^2)$ independence queries was known.

In this work, we break the quadratic barrier with a randomized algorithm guaranteeing $\tO(n^{9/5})$ independence queries with high probability, and a deterministic algorithm guaranteeing $\tO(n^{11/6})$ independence queries.
Our key insight is simple and fast algorithms to solve a graph reachability problem that arose in the standard augmenting path framework [Edmonds 1968]. Combining this with previous exact and approximation algorithms leads to our results.



\end{abstract}

%% file: intro.tex
\section{Introduction} \label{sec:intro}

\paragraph{Matroid intersection.}  

The matroid intersection problem is a fundamental combinatorial optimization problem that has been studied for over half a century. A wide variety of prominent optimization problems, such as bipartite matching, finding an arborescence, finding a rainbow spanning tree, and spanning tree packing, can be modeled as matroid intersection problems \cite[Chapter 41]{schrijver2003}. Hence the matroid intersection problem is a natural avenue to study all of these problems simultaneously.

Formally, a matroid is  defined by the the tuple $\cM= (V, \cI)$ where $V$ is a finite set of size $n$, called the \textit{ground set}, and $\cI \subseteq 2^V$ is a family of subsets of $V$, known as the \textit{independent sets}, that satisfy two properties: (i) $\cI$ is \textit{downward closed}, i.e., all subsets of any set in $\cI$ are also in $\cI$, and (ii) for any two sets $A, B \in \cI$ with $|A| < |B|$, there is an element $v \in B \setminus A$ such that $A \cup \{v\} \in \cI$, i.e., $A$ can be \textit{extended} by an element in $B$. Given two such matroids $\cM_1 = (V, \cI_1)$ and $\cM_2 = (V, \cI_2)$ defined over the same ground set $V$, the matroid intersection problem asks to output the largest common independent set $S \in \cI_1 \cap\cI_2$. The size of such a set is called {\em rank} and is denoted by $r$. 

The classic version of this problem that has been studied since the 1960s assumes \textit{independence query} access to the matroids: Given a matroid $\cM$, an independence oracle takes a set $S \subseteq V$ as input and outputs a single boolean bit depending on whether $S \in \cI$ or not, i.e., it outputs 1 iff $S \in \cI$. The matroid intersection problem assumes the existence of two such independence oracles, one for each matroid. 
The goal is to design an efficient algorithm in order to minimize the number of such oracle accesses, i.e., to minimize the independence query complexity of the matroid intersection problem. This is the version of the problem that we study in this work. 
Note that a more powerful query model called {\em rank query} has been recently studied in \cite{lee2015faster, chakrabarty2019faster}. We do {\em not} consider such model.




\vspace{-7pt}
\paragraph{Previous work.} Starting with the work of Edmonds in the 1960s,
algorithms with polynomial query complexity for matroid intersection have been studied
\cite{edmonds1968matroid,edmonds1970submodular, aignerD, Lawler75, edmonds1979matroid,cunningham1986improved, lee2015faster,nguyen2019note,chakrabarty2019faster}.
In 1986, Cunningham \cite{cunningham1986improved} designed an algorithm with query complexity $O(nr^{1.5})$ based on the ``blocking flow'' ideas similar to Hopcroft-Karp's bipartite-matching algorithm or Dinic's maximum flow algorithm. This was the best query algorithm for the matroid intersection problem for close to three decades until the recent works of Nguy\~{\^e}n~\cite{nguyen2019note} and Chakrabarty-Lee-Sidford-Singla-Wong ~\cite{chakrabarty2019faster} who independently showed that Cunningham's algorithm
can be implemented using only $\tO(nr)$ independence queries. In a separate line of work, Lee-Sidford-Wong \cite{lee2015faster} proposed a cutting plane algorithm using $\tO(n^2)$ independence queries. When $r$ is sublinear in $n$, the result of \cite{chakrabarty2019faster, nguyen2019note} provides faster (subquadratic) algorithm than that of \cite{lee2015faster}, but for linear $r$ (i.e., $r \approx n$), all of these results are stuck at query complexity of $\tO(n^2)$. This is known as the \textit{quadratic barrier} \cite{chakrabarty2019faster}. A natural question is whether this barrier can be broken \cite[Conjecture 13]{lee2015faster}. 


The only previous progress towards breaking this barrier is by \cite{chakrabarty2019faster} and falls under the following two categories. Either we need to assume the more powerful {\em rank} oracle model where \cite{chakrabarty2019faster} provides a $\tO(n^{1.5})$-time algorithm. Or, we solve an {\em approximate} version of the matroid intersection problem, where \cite{chakrabarty2019faster} provides an algorithm with $\tO(n^{1.5}/\eps^{1.5})$ complexity for $(1-\eps)$-approximately solving the matroid intersection problem in the independence oracle model. Breaking the quadratic barrier with an {\em exact} algorithm in the {\em independence} query model remains open.


\vspace{-7pt}
\paragraph{Our results.} 
We break the quadratic barrier with both deterministic and randomized algorithms:  


\begin{theorem}[Details in \cref{thm:mi-rand,thm:mi-det}]
\label{thm:mi-overview}
Matroid Intersection can be solved by
\begin{itemize}[noitemsep]
\item a deterministic algorithm taking $\tO(n^{11/6})$ independence queries, and
\item a randomized (Las Vegas) algorithm taking $\tO(n^{9/5})$ independence queries with high probability.
\end{itemize}
\end{theorem}

By high probability, we mean probability of at least $1 - 1/n^c$ for an arbitrarily large constant~$c$. While we only focus on the query complexity in this paper, we note that the time complexities of our algorithms
are dominated by the independence oracle queries.
That is, our deterministic and randomized algorithms have time complexity $\tO(n^{11/6}\mathcal{T}_{\text{ind}})$ and $\tO(n^{9/5}\mathcal{T}_{\text{ind}})$ respectively, where
$\mathcal{T}_{\text{ind}}$ denotes the maximum time taken by an oracle to answer an independence query.

\paragraph{Technical overview.}
Below we explain the key insights of our algorithms which are fast algorithms to solve a graph problem called {\em reachability} and a simple way to combine our algorithms with the existing exact and approximation algorithms to break the quadratic barrier.




\medskip\noindent{\em Reachability problem:} In this problem, there is a directed bipartite graph $G$ on $n$ vertices with
bi-partition $(S\cup \{s,t\}, \bar{S})$.
We want to determine whether a directed $(s,t)$-path exists
in $G$. We know the vertices of $G$, but not the edge set $E$ of $G$. We are allowed to ask the following two types of {\em neighborhood queries}: 

\begin{enumerate}[noitemsep]
\item \label{itm:query-1} Out-neighbor query: Given $v\in \bar{S}$ and $X\subseteq S\cup \{s,t\}$, does there exists
an edge from $v$ to some vertex in $X$?
\item \label{itm:query-2} In-neighbor query: Given $v\in \bar{S}$ and $X\subseteq S\cup \{s,t\}$, does there exists
an edge from some vertex in $X$ to $v$?
\end{enumerate}


In other words, we can ask an oracle if a ``right vertex'' $v\in \bar S$ has an edge to or from a set $X$ of ``left vertices''. This problem arose as a subroutine of previous matroid intersection algorithms that are based on finding  augmenting paths \cite{aignerD,Lawler75,cunningham1986improved,chakrabarty2019faster,nguyen2019note}. Naively, we can solve this problem with quadratic ($O(n^2)$) queries: find all edges of $G$ by making a query for all possible $O(n^2)$ pairs of vertices. Cunningham \cite{cunningham1986improved} used this algorithm in his framework to solve the matroid intersection problem with $O(nr^{1.5})$ queries. Recent results by \cite{chakrabarty2019faster,nguyen2019note} solved the reachability problem with $\tO(nd)$ queries, where $d$ is the distance between $s$ and $t$ in $G$, essentially by simulating the breadth-first search process. Plugging these algorithms into Cunningham's framework leads to algorithms for the matroid intersection with $\tO(nr)$ queries. When $d$ is large, the algorithms of \cite{chakrabarty2019faster,nguyen2019note} still need $\tilde\Theta(n^2)$ queries to solve the reachability problem. It is not clear how to solve this problem with a subquadratic number of queries. The key component of our algorithms is subquadratic-query algorithms for the reachability problem:

%


\begin{theorem} [Details in \cref{thm:aug,thm:aug-det}]
\label{thm:aug-overview}
The reachability problem can be solved by
\begin{itemize}[noitemsep]
\item a deterministic algorithm that takes  $\tO(n^{5/3})$ queries, and
\item a randomized (Las Vegas) algorithm that takes $\tO(n\sqrt{n})$ queries with high probability.
\end{itemize}
\end{theorem}

Plugging \Cref{thm:aug-overview} into standard frameworks such as Cunningham's does not directly lead us to a subquadratic-query algorithm for matroid intersection. Our second insight is a simple way to combine algorithms for the reachability problem with the exact and approximation algorithms of \cite{chakrabarty2019faster} to achieve the following theorem. 



\begin{theorem}[Details in \cref{thm:mi-given-aug}]\label{thm:intro:connection}
If there is an algorithm $\cal A$ that solves the reachability problem with ${\cal T}$ queries, then there is an algorithm $\cal B$ that solves the matroid intersection problem with $\tO(n^{9/5} + n\sqrt{\mathcal{T}})$ independence queries. If ${\cal A}$ is deterministic, then ${\cal B}$ is also deterministic. 
%
\end{theorem}

\Cref{thm:aug-overview,thm:intro:connection} immediately lead to \Cref{thm:mi-overview}. 
We provide proof ideas of \Cref{thm:aug-overview,thm:intro:connection} in the subsections below.

\subsection{Proof idea for Theorem~\ref{thm:aug-overview}: Algorithm for the reachability problem}\label{sec:intro:aug-overview}
Before mentioning an overview of the algorithm for solving the reachability problem, we briefly mention what makes this problem hard. Note that if we discover that some $v\in \bar{S}$ is reachable from $s$, we can find
all out-neighbors of $v$ in $(S\cup\{s,t\})$ in $O(\log n)$ queries per such neighbor. We do this by using a binary search with out-neighbor queries,
halving the size of the set of potential out-neighbors of $v$ in each step. However, when we discover that some $v\in S$ is reachable from $s$, we cannot use the same binary-search trick to efficiently find the out-neighbors of $v$ due to the {\em asymmetry} of the allowed queries, where we can make queries only for vertices $v\in \bar{S}$. Such asymmetry makes it hard to efficiently apply a standard $(s,t)$-reachability algorithm (such as breadth-first search) on the graph.\footnote{In contrast, in the {\em symmetric} case where in- and out-neighbor queries can be made for {\em every} vertex (and not just $v\in\bar S$), we can solve the reachability problem with  $\tO(n)$ queries. This requires a simple breadth-first search starting from $s$ where we discover neighbors using binary search.}

Both our randomized and deterministic algorithms for the reachability problem follow the same framework below, where we partition vertices in $\bar S$ into {\em heavy and light} vertices and find vertices that can reach some heavy vertices. Our randomized and deterministic algorithms differ in how they determine whether a vertex is heavy or light. 

\paragraph{Heavy/Light vertices.} 
Our reachability algorithms run in phases and keep track of a set of vertices that are reachable from the source vertex $s$, denoted by $F$ (for ``found''). We can assume that {\em $F$ contains all out-neighbors of vertices in $F \cap \bar S$},
because we can find these out-neighbors very efficiently by doing binary-search that makes $\tO(1)$ queries per out-neighbor. In each phase, the algorithm either 
\begin{itemize}[noitemsep]
    \item[(a)] increases the size of $F$ by an additive factor of at least $h$ for some parameter $h$ (we use either  $h = \sqrt{n}$ or $h = n^{1/3}$), or
    \item[(b)] returns whether there is an $(s,t)$ path. 
\end{itemize}
Hence, in total, there are at most $\frac{n}{h}$ many phases.
%
%
%
To this end, for every $v\in \bar S$, we say that $v$ is {\em $F$-heavy} if either
\begin{itemize}[noitemsep]
    \item[(h1)] $v$ has at least $h$ out-neighbors to $S\setminus F$, or 
    \item[(h2)] there is an edge from $v$ to $t$.
\end{itemize}
If $v\in\bar S$ is not $F$-heavy, we say that it is {\em $F$-light}. We omit $F$ when it is clear from the context.
We emphasize that the notion of heavy and light applies only to vertices in $\bar S$.\danupon{TO DO: Say that we need $n$ queries per phase?} 
Two tasks that remain are how to determine if a vertex is heavy or light, and how to use this to achieve (a) or (b). 

\paragraph{Heavy vertex reachability.} First, we show how to achieve (a) or (b). We assume for now that we know which vertices in $\bar S$ are heavy or light. We can also assume that we know all out-going edges of all light vertices (e.g. black edges in \Cref{fig:rev-bfs}); this requires $\tilde O(nh)$ queries over all phases.
Our main component is to determine a set of vertices that can reach {\em some} heavy vertex. (Heavy vertices are always in such set.)
We can do this with $\tilde O(n)$ queries essentially by simulating a breadth-first search process {\em reversely} from heavy vertices. This process leaves us with subtrees rooted at the heavy vertices with edges pointed to the roots; see \Cref{fig:rev-bfs} for an example. The actual algorithm is quite simple and can be found in \Cref{sec:augmentation}. 
%

%
Once all vertices that can reach some heavy vertices are found, we end up in one of the following situations: 
\begin{itemize}[noitemsep]
    \item Some vertex in $F$ can reach a heavy vertex $v$ satisfying (h2). In this case, we know immediately that $s$ can reach $t$ via $v$. 
    \item Some vertex in $F$ can reach a heavy vertex $v$ satisfying (h1). In this case, we query and add all out-neighbors of $v$ in $S\setminus F$ (taking $\tO(n)$ queries). This adds at least $h$ vertices to $F$ as desired.
    \item No vertices in $F$ can reach any heavy vertex. In this case, we conclude that $s$ does not reach $t$: to be able to reach $t$, $s$ must be able to reach some vertex that points to $t$ (and thus is heavy). 
\end{itemize}

\begin{figure}
\centering
\setlength{\belowcaptionskip}{6pt}
\includegraphics[scale=0.8]{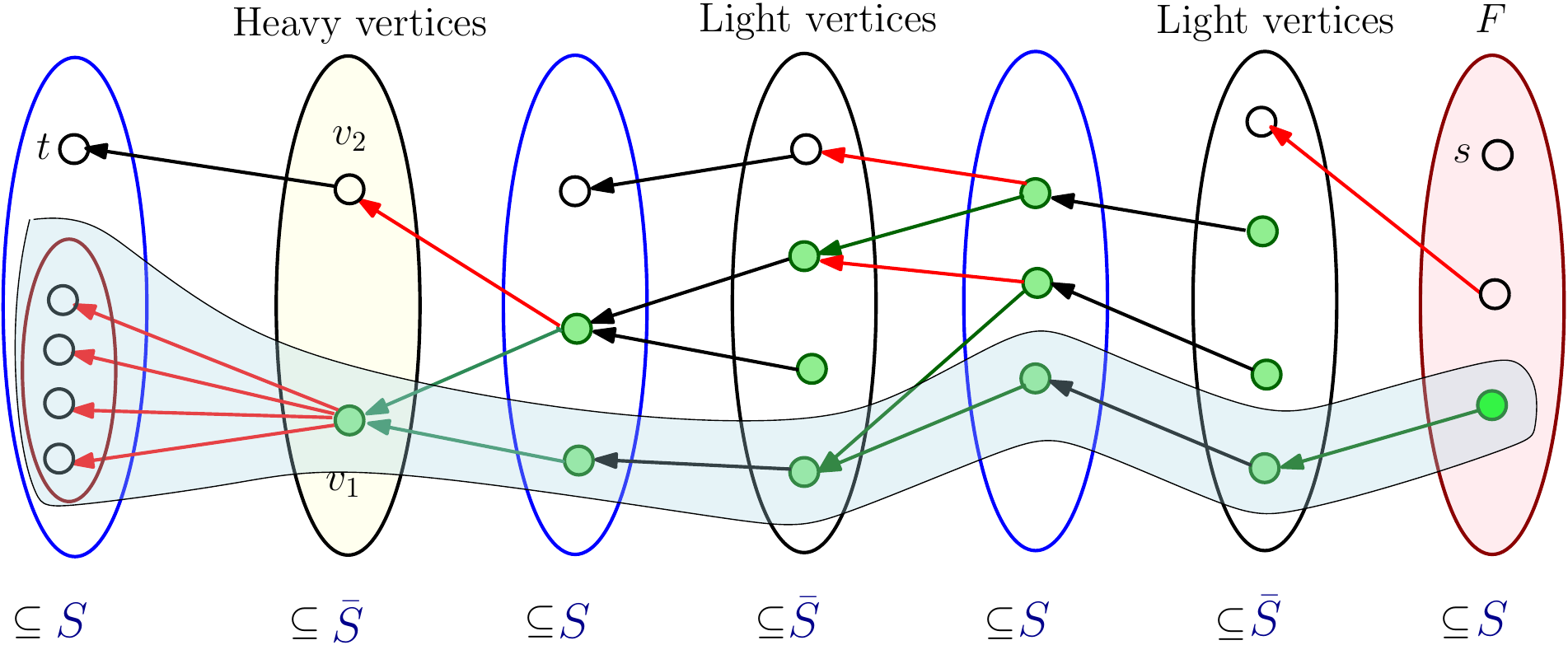}
 \caption{\small Example of the reverse BFS process to compute heavy vertex reachability. The vertices from $S$ and $\bar S$ occur at alternate layers. The black edges (out-edges of light nodes) are known a priori. The \textcolor{green!50!black}{green} edges are traversed in the reverse BFS procedure whereas the \textcolor{red}{red} edges are not traversed in the reverse BFS. The \textcolor{green!50!black}{green} vertices are discovered in the reverse BFS. The green vertices form a tree rooted at $v_1$. Vertex $v_1$ is heavy because of its large out-degree. Vertex $v_2$ is heavy because $t$ is its out-neighbor. The path from $F$ to $v_1$ and the out-neighbors of $v_1$ are highlighted in light-blue, which is added to $F$ after the reverse BFS. Note that, even though $v_2$ is reachable from $F$, the path from $F$ to $v_2$ is not discovered, and hence the algorithm moves to the next iteration. 
 }
\label{fig:rev-bfs} 
\end{figure}

\paragraph{Heavy/light categorization.} Again, this is where our randomized and deterministic algorithms differ. With randomness, we can use random sampling to approximate the out-degree of every vertex in $\bar S$ and find all out-going edges of vertices that are potentially light. This takes $\tilde O(nh+n^2/h)$ queries over all phases. 
For the deterministic algorithm, a naive idea is to maintain, for every $v\in\bar S$, up to $h$ out-going neighbors of $v$ in $S\setminus F$. The challenge is that when these neighbors are included in $F$, we have to find new neighbors. By carefully picking these neighbors, we can argue that in total only $\tO(n\sqrt{nh})$ queries are needed over all phases.

\paragraph{Summary.}
In total, in addition to the categorization of \emph{heavy}/\emph{light} vertices, we use
$\tO(n)$ queries to solve the heavy vertex reachability problem in each of the $\tO(n/h)$ phases. We also need $\tO(nh)$ queries over all phases to find at most $h$ out-neighbors in $S\setminus F$ of light vertices. So, in total,
our algorithm uses $\tO(\frac{n^2}{h} + nh)$ queries plus the number of queries needed
for the categorization, which is $\tO(\frac{n^2}{h} + nh)$ for randomized
and $\tO(n\sqrt{nh})$ for deterministic algorithms.

\subsection{Proof idea of Theorem~\ref{thm:intro:connection}: From reachability to matroid intersection}\label{sec:overview:transform}

The standard connection between the matroid intersection problem and the reachability problem that is exploited by most combinatorial algorithms \cite{aignerD,Lawler75,cunningham1986improved,chakrabarty2019faster,nguyen2019note} is based on finding
augmenting paths in what is called the \emph{exchange graph}.
Given a common independent set $S$ of the two matroids over common ground set $V$,
the \emph{exchange graph}~$G(S)$ is a directed bipartite graph over vertex set $V \cup \{s,t\}$ as in the
reachability problem above with $\bar S = V \setminus S$. The edges of the exchange graph are defined to ensure the following property: Finding an $(s,t)$-path in the exchange graph amounts to augmenting $S$,
i.e.~finding a new common
independent with a bigger size. Conversely,
if no $(s,t)$-path exists in the exchange graph, it is known that $S$ is of maximum
cardinality and, hence, $S$ can be output as the answer to the matroid intersection problem. Thus the problem of \textit{augmentation} in the exchange graph can be reduced to the \textit{reachability} problem where the neighborhood queries in the reachability problem correspond to the queries to the matroid oracles.\footnote{The independence queries are more powerful
than the neighborhood queries, but we are only interested in the neighborhood queries in our algorithm for the reachability problem.}


Let us suppose that we can solve the \emph{reachability} problem using $\mathcal{T}$
queries. An immediate and straightforward way of using this subroutine to solve matroid intersection is the following: Call this subroutine iteratively to find augmenting paths to augment along
in the exchange graph, thereby increasing the size of the common independent set by one in each iteration. As the size of the largest common independent set is $r$, we need to perform $r$ augmentations in total. This leads to an algorithm solving matroid intersection using $O(r\mathcal{T})$ independence queries.

To improve upon this, we avoid doing the majority of the augmentations by starting
with a good approximation of the largest common independent set.
We use the recent subquadratic $(1-\eps)$-approximation algorithm of \cite[Section~6]{chakrabarty2019faster} that uses
$\tO(n^{1.5}/\eps^{1.5})$ independence queries
to obtain a common independent set of size at least $r-\eps r$.
Once we obtain a common independent set with such approximation guarantee, we only need to perform
an additional $\eps r$ augmentations.
This is still not good enough to obtain a subquadratic matroid intersection algorithm
when combined with our efficient algorithms for the reachability problem from
\cref{thm:aug-overview}.

The final observation we make, is that for small $\eps = o(n^{-1/5})$, we can combine
the $\tO(n^{1.5}/\eps^{1.5})$ approximation algorithm of \cite{chakrabarty2019faster} with an efficient implementation of Cunningham's algorithm (as in \cite{chakrabarty2019faster,nguyen2019note}) to obtain a $(1-\eps)$-approximation algorithm for matroid intersection using $\tO(n^{9/5} + n/\eps)$ queries. This has a slightly better complexity than just running the approximation algorithm of \cite{chakrabarty2019faster}. The idea is to first run the
the $\tO(n^{1.5}/\eps'^{1.5})$ approximation algorithm with $\eps' \approx n^{-1/5}$,
and then run the Cunningham-style algorithm until the distance between $s$ and $t$
in the exchange graph becomes at least $\Theta(1/\eps)$. 

Our final algorithm is then:
\begin{enumerate}[noitemsep]
\item Run
the $\tO(n^{1.5}/\eps^{1.5})$-query $(1-\eps)$-approximation algorithm from \cite[Section~6]{chakrabarty2019faster} with $\eps = n^{-1/5}$ to obtain a common
independent set $S$ of size at least $r-n^{4/5}$.
This step takes $\tO(n^{9/5})$ queries.
\item Starting with $S$, run the Cunningham-style algorithm as implemented by \cite[Section~5]{chakrabarty2019faster} until the $(s,t)$-distance is at least
$\sqrt{\mathcal{T}}$ to obtain a common independent set of size at least $r-O(n/\sqrt{\mathcal{T}})$.
This step takes $\tO(n(r-|S|) + n\sqrt{\mathcal{T}}) = \tO(n^{9/5} + n\sqrt{\mathcal{T}})$ queries.
\item For the remaining $O(n/\sqrt{\mathcal{T}})$ augmentations, find augmenting paths
one by one by solving the reachability problem.
This step takes $\tO(n\sqrt{\mathcal{T}})$ queries.
\end{enumerate}
Hence we obtain a matroid intersection algorithm which uses
$\tO(n^{9/5}+n\sqrt{\mathcal{T}})$ independence queries, as in \cref{thm:intro:connection}.

\subsection{Organization} 
We start with the necessary preliminaries in \Cref{sec:prelims}. In \Cref{sec:augmentation}, we provide the subquadratic deterministic and randomized algorithms for augmentation. Finally, in \Cref{sec:main-algo}, we combine these algorithms for augmentation with existing algorithms to obtain subquadratic deterministic and randomized algorithms for matroid intersection. In \Cref{sec:augmentation}, we skip the description of an important subroutine called the heavy/light categorization. We devote \Cref{sec:heavy-light} for details of this subroutine.